\documentclass[11pt]{article}

\usepackage[english]{babel}

\usepackage[letterpaper,top=2cm,bottom=2cm,left=3cm,right=3cm,marginparwidth=1.75cm]{geometry}

\usepackage{amsmath}
\usepackage{amssymb}
\usepackage{latexsym}
\usepackage{mathtools}
\usepackage{graphicx}
\usepackage{graphicx,epstopdf}
\usepackage{graphicx,rotating,booktabs}
\usepackage{natbib}
\usepackage{multirow}
\usepackage{subfigure}
\usepackage{xcolor}
\usepackage[colorlinks=true, allcolors=blue]{hyperref}
\usepackage{tikz} 
\usetikzlibrary{positioning,fit,calc}
\usetikzlibrary{shapes,arrows}
\tikzset{block/.style={draw, thick, text width=2cm ,minimum height=1.3cm, align=center}, line/.style={-latex}}  
\usepackage{authblk} 
\usepackage{float}
\usepackage{gensymb}
\usepackage{algorithm,algcompatible}
\usepackage{longtable}

\usepackage{amsthm}
\newtheorem{theorem}{Theorem}
\newtheorem{corollary}{Corollary}[theorem]

\newtheorem{assumption}{Assumption} 

\newcommand{\bc}{\begin{center}}
\newcommand{\ec}{\end{center}}
\newcommand{\bi}{\begin{itemize}}
\newcommand{\ei}{\end{itemize}}
\newcommand{\be}{\begin{enumerate}}
\newcommand{\ee}{\end{enumerate}}
\newcommand{\bd}{\begin{description}}
\newcommand{\ed}{\end{description}}
\newcommand{\bs}{\begin{single}}
\newcommand{\es}{\end{single}}

\algnewcommand\INPUT{\item[\textbf{Input:}]}%
\algnewcommand\OUTPUT{\item[\textbf{Output:}]}%

\title{Generalized Bayesian Inference using the Bayesian Bootstrap for Survival Models} 
\author[1]{K Shuvo Bakar \thanks{\href{mailto:shuvo.bakar@sydney.edu.au}{shuvo.bakar@sydney.edu.au}}}
\author[1,2]{Armando Teixeira-Pinto}
\affil[1]{Biostatistics, Sydney School of Public Health, The University of Sydney, Camperdown, NSW-2050, Australia}
\affil[2]{Centre for Kidney Research, Kids Research Institute, The Children’s Hospital at Westmead, NSW-2145, Australia}

\date{}

\begin{document}

\maketitle

\begin{abstract}

Survival inference often requires uncertainty quantification under censoring and limited sample sizes, while prior information may be available but difficult to incorporate without specifying a full likelihood. Likelihood-based Bayesian survival methods provide prior-informed inference but may be sensitive to distributional assumptions and computationally demanding for complex survival models. The Bayesian bootstrap provides a likelihood-free, nonparametric approach to uncertainty quantification, but by itself does not provide a general mechanism for incorporating priors on model parameters. We propose a general Bayesian method for survival models that combines the generalized Bayesian (Gibbs) updating with Bayesian bootstrap. The Bayesian bootstrap generates a distribution over a target survival estimator through Dirichlet weights, while generalized Bayesian updating incorporates parameter-specific prior information through a loss function. The resulting posterior provides a flexible alternative to likelihood-based Bayesian inference and can be applied to a broad class of survival estimators. We develop the method for the Cox proportional hazards model and obtain posterior inference for regression coefficients and hazard ratios. Simulation studies demonstrate uncertainty quantification and prior-data learning across varying sample sizes. An application to right-censored survival data illustrates its practical utility. The methodology is implemented in the open-source \texttt{R} package \texttt{BayesBoots}.
\end{abstract}


\section{Introduction}\label{sec:intro}

Survival analysis is a fundamental statistical method for studying time-to-event outcomes in medicine, epidemiology, reliability engineering, and public health~\citep{andersen2012statistical}. Often survival data includes censoring, where event times are only partially observed, making uncertainty quantification and robust inference difficult, particularly in small or heterogeneous datasets~\citep{klein2003survival}. For a right-censored survival data, we define $\mathcal{D} = \{(T_i, \delta_i, X_i)\}_{i=1}^n$, where $T_i$ denotes the observed time, $\delta_i \in \{0,1\}$ indicates whether the event of interest is observed, and $X_i \in \mathbb{R}^p$ are covariates.
Classical survival methods such as the Cox proportional hazards model remain the dominant tools because of their interpretability, robustness, and easy modeling assumptions. These approaches are primarily frequentist and rely largely on asymptotic uncertainty estimates~\citep{fleming2013counting}. In many modern applications, especially biomedical studies, there is increasing need for methods that can provide richer uncertainty characterization while incorporating prior knowledge~\citep{spiegelhalter2004bayesian}.

\subsection{Problem setup}

Bayesian survival analysis addresses some of these limitations of uncertainty characterization by introducing prior distributions~\citep{bartovs2022informed}. Existing Bayesian approaches typically rely on fully specified probabilistic likelihood models, including parametric survival distributions, Bayesian Cox models, and Bayesian nonparametric constructions such as Dirichlet process mixtures~\citep{ibrahim2001bayesian}. Although flexible, Bayesian survival models often require strong distributional assumptions, can be computationally demanding due to posterior Markov chain Monte Carlo (MCMC) sampling, and may be sensitive to model misspecification, particularly in the presence of complex hazard functions, heterogeneity, or substantial censoring~\citep{hanson2006modeling,ibrahim2001bayesian}. 
To explain more, for example, one important consideration in Bayesian Cox model is the treatment of the baseline hazard function. In the classical Cox, the baseline hazard is left unspecified, allowing estimation of model parameters through the partial likelihood without requiring a parametric assumption about the underlying hazard function~\citep{cox1972regression}. Whereas, Bayesian implementations require the baseline hazard to be modeled explicitly in order to construct the full likelihood and posterior distribution. Some common approaches include piecewise constant hazard models, where the baseline hazard is approximated using flexible spline-based approaches that model the baseline hazard using smooth functions~\citep{murray2015flexible,ibrahim2001bayesian}, and stochastic process priors such as Weibull or Gamma~\citep{lee2016hierarchical} that provide flexible Bayesian representations of the cumulative baseline hazard~\citep{hjort1990nonparametric}. However, introducing a prior structure for the baseline hazard creates additional modelling choices, and posterior inference may be sensitive to the selected baseline hazard specification.
Another challenge relates to handling tied event times. In conventional Cox regression, tied survival times are accommodated through well-established approximations to the partial likelihood, including the Breslow and Efron methods, which are computationally efficient and widely implemented~\citep{breslow1974covariance,efron1977efficiency}. In Bayesian Cox models, however, the full likelihood formulation often requires explicit treatment of the event process, and tied observations may introduce additional complexity~\citep{ibrahim2001bayesian,gamerman1991dynamic}. 
Scalability in the computation is also a factor to consider. Although Bayesian Cox models provide a flexible framework for uncertainty quantification and incorporation of prior knowledge, computational demands can become substantial for very large survival datasets. This challenge is primarily related to MCMC sampling method, which sometimes need substantial computational time and careful convergence assessment. For large-scale applications, alternative computational strategies, including variational inference and other approximate Bayesian methods, have been explored to improve scalability while maintaining Bayesian inference properties~\citep{blei2017variational,ranganath2014black}. While these approaches improve computational efficiency, they may introduce additional assumptions and require validation against more computationally intensive methods.

\subsection{Theoretical contributions and underlying concept}

These limitations motivate an alternative approach to get inference that is based on the empirical stability of model estimators. 
Our paper contributes to the literature on the foundational statistical properties of Bayesian inference using the Bayesian bootstrap, a concept developed in the early 1980s by~\citet{rubin1981bayesian} following the bootstrap method proposed by~\citet{efron1979bootstrap}.
We give theorems to establish the relationship between Bayesian inferences using conventional prior-likelihood and using Bayesian bootstrap-based generalized inference through Bernstein–von Mises arguments~\citep{van1998asymptotic}. We propose that under non-negligible prior settings the Bayesian bootstrap-based generalized inference is equivalent to the Bayesian inference. 


The Bayesian bootstrap offers a natural mechanism~\citep{lo1993bayesian}, where by assigning random Dirichlet weights to the observed data, it creates a distribution over quantities such as regression coefficients without imposing parametric assumptions on the underlying model. The resulting uncertainty is entirely data-driven. However, because the Dirichlet weighting scheme is centered on the empirical distribution, it does not readily allow for the incorporation of external or subject-matter prior information beyond what is implicitly encoded in the data.
Building on this perspective, we propose a Bayesian bootstrap framework for regression models that explicitly enables the incorporation of prior information. Compared to standard Bayesian bootstrap constructions, where uncertainty is driven solely by Dirichlet perturbations of the empirical measure from Monte Carlo (MC) samples, our approach treats Bayesian bootstrap distributions as baseline inferential objects~\citep{newton2021weighted,ballante2025generalized}, which captures data-driven uncertainty of the empirical distribution.
Model parameter-specific prior information is then incorporated through generalized Bayesian updating~\citep{bissiri2016general,nie2023deep}, also known as Gibbs inference~\citep{overstall2025gibbs}, where the bootstrap MC samples are reweighted using a loss function together with a prior penalty. Adding this step in our problem context introduces prior-guided posterior updating, effectively transforming the bootstrap distribution into a generalized posterior~\citep{lyddon2019general}. 

\subsection{Survival modeling setup}

We implement the proposed method for survival analysis framework, i.e., for modeling time-to-event data using survival models such as the Cox proportional hazard~\citep{hosmer2008applied}. 
Unlike Bayesian survival models, which require specification of a full probabilistic model and derive inference directly from an explicit likelihood, we use the proposed method to decouple uncertainty quantification, first through the Bayesian bootstrap, and then refining via loss-based Gibbs reweighting using model parameter specific prior information, which measures variability in the observed risk sets and censoring structure. 
This provides an  attractive alternative to fully Bayesian inference, particularly in large-scale or complex survival settings~\citep{minsker2014scalable}.

The proposed framework is general and not restricted to any specific survival model. It can be applied to a broad range of estimators and survival functions, including hazard functions, cumulative incidence functions, and regression parameters from semiparametric or parametric models. As a concrete example, in this paper, we use the Cox proportional hazards model~\citep{hosmer2008applied}. This yields a flexible and robust alternative to Bayesian Cox regression while retaining interpretability in terms of hazard ratios. We have also developed an \texttt{R} package \texttt{BayesBoots} to implement the proposed method for Cox regression model, which is currently available to download from \href{https://github.com/ksbakar/BayesBoots}{GitHub} repository.

\subsection{Organization}

The rest of the paper is organized as follows: in Section~\ref{sec:methods}, we develop the method for Cox model, proposing an algorithm to implement the method. Section~\ref{sec:BB_Cox_Bayes_relation} provides asymptotic assumptions, theorems and related proofs to compare the proposed method with the conventional Bayesian inference for Cox model. In Section~\ref{sec:sim}, we conduct simulation studies and then Section~\ref{sec:ex} provides data example related to kidney allograft failure. Finally, in Section~\ref{sec:con} we include discussions and conclusions. 

\section{Generalized Bayesian-bootstrap Cox model (BB-Cox)}\label{sec:methods}

In this section, we develop the generalized Bayesian-bootstrap Cox model (BB-Cox) and establish its relationship with Bayesian Cox regression. The Cox proportional hazards model provides a natural setting for this development because it permits covariate effects to be estimated without specification of the baseline hazard and obtains inference through the partial likelihood~\citep{cox1972regression}. 
Let us first define the Cox proportional hazards model as: $h(t\mid X_i) = h_0(t)\exp(X_i^\top\beta)$, where $h_0(t)$ is an unspecified baseline hazard function and $\beta\in\mathbb{R}^p$ is the vector of regression coefficients~\citep{cox1972regression,hosmer2008applied}. Let $\mathcal{R}(T_i)$ denote the risk set immediately before the observed failure time $T_i$. For notational consistency, we denote the Cox partial log-likelihood by
\begin{eqnarray}
\ell_n(\beta) = \sum_{i:\delta_i=1} \left[X_i^\top\beta-\log\left\{\sum_{j\in\mathcal{R}(T_i)}\exp(X_j^\top\beta)\right\}\right],
\label{eq:cox_loglik}
\end{eqnarray}
and its negative by $\mathcal{L}_n(\beta)=-\ell_n(\beta)$. The ordinary Cox partial likelihood estimator is therefore $\hat{\beta} = \arg\max_{\beta}\ell_n(\beta)$, or  $\hat{\beta} = \arg\min_{\beta}\mathcal{L}_n(\beta)$.
The distinction between $\ell_n(\beta)$ and $\mathcal{L}_n(\beta)$ is useful throughout the paper: $\ell_n(\beta)$ denotes the log partial likelihood when discussing Bayesian posterior distributions, whereas $\mathcal{L}_n(\beta)$ denotes the corresponding negative objective when defining the weighted Cox estimators.

\subsection{Bayesian-bootstrap distribution}
\label{sec:bb}

The proposed method combines Bayesian-bootstrap perturbations~\citep{rubin1981bayesian,lo1987large} of the Cox estimator with generalized Bayesian updating. The Bayesian bootstrap provides a nonparametric mechanism for propagating empirical-distribution uncertainty without requiring a parametric model for the baseline hazard~\citep{lo1993bayesian,newton1994approximate}.
For each bootstrap replication $m$, where $m=1,\ldots,M$; we independently generate weights: $w_1^{(m)},\ldots,w_n^{(m)} \sim \text{Dirichlet}(1,\ldots,1)$.
These weights define a perturbed empirical distribution. The corresponding weighted negative Cox partial log-likelihood is
\begin{eqnarray}
\mathcal{L}^{(m)}_n(\beta) = -\sum_{i:\delta_i=1} w_i^{(m)}\left[X_i^\top\beta-\log\left\{\sum_{j\in\mathcal{R}(T_i)}w_j^{(m)}\exp(X_j^\top\beta)\right\}\right],
\label{eq:weighted_cox_loss}
\end{eqnarray}
and the associated Bayesian-bootstrap Cox estimator is $\beta^{(m)} = \arg\min_{\beta}\mathcal{L}^{(m)}_n(\beta)$. 
Repeating this procedure for $m=1,\ldots,M$ produces $\{\beta^{(m)}\}_{m=1}^{M}$,
which represents the empirical Bayesian-bootstrap distribution of the Cox regression coefficient~\citep{chen2012monte}. We define
\begin{eqnarray}
Q_{\text{BB}}(\beta) = \frac{1}{M}\sum_{m=1}^{M}\delta_{\beta^{(m)}}(\beta),
\label{eq:qbb}
\end{eqnarray}
where $\delta_{\beta^{(m)}}$ denotes the Dirac probability measure concentrated at $\beta^{(m)}$. The Bayesian bootstrap provides a data-driven representation of sampling uncertainty associated with the Cox estimator~\citep{kim2003bayesian,lo1993bayesian}, although its use avoids placing a parametric prior directly on the baseline hazard or the regression parameter.

\subsection{Gibbs updating}

The second stage applies generalized Bayesian updating to the empirical measure $Q_{\text{BB}}$. Generalized Bayesian inference constructs probability distributions by exponentiating a loss function rather than requiring the loss to arise from a fully specified likelihood~\citep{bissiri2016general}. This provides a natural framework for incorporating prior information while retaining the empirical uncertainty representation generated by the Bayesian bootstrap. We now define a loss function on the coefficient space by $L(\beta,\beta') = \sum_{k=1}^{p}(\beta_k-\beta_k')^2=\|\beta-\beta'\|_2^2$.
For each bootstrap coefficient $\beta^{(m)}$, the corresponding empirical risk is
\begin{eqnarray}
\widetilde{\ell}^{(m)} = \sum_{j=1}^{M}L\!\left(\beta^{(m)},\beta^{(j)}\right).
\label{eq:empirical_risk}
\end{eqnarray}
Thus, $\widetilde{\ell}^{(m)}$ measures the discrepancy between the $m$th bootstrap coefficient and the ensemble of bootstrap coefficient estimates.
The Gibbs normalized weight assigned to $\beta^{(m)}$ is $\pi_m \propto \exp\{-\lambda\widetilde{\ell}^{(m)}-\Omega(\beta^{(m)})\}$,
where $\lambda>0$ is the Gibbs learning parameter and $\Omega(\beta)$ is a prior-induced penalty. We connect the penalty to a Bayesian Cox prior through $\pi(\beta) \propto \exp\{-\Omega(\beta)\}$.
For example, a Gaussian prior on $\beta$ corresponds to a quadratic penalty. Hence, the resulting BB-Cox generalized posterior is
\begin{eqnarray}
Q_{\text{BB-Cox}}(\beta) = \sum_{m=1}^{M} \pi_m \delta_{\beta^{(m)}}(\beta).
\label{eq:bb-cox}
\end{eqnarray}
Thus, BB-Cox does not generate a new continuous parameter distribution independently of the bootstrap samples. Instead, it redistributes probability mass over the Bayesian-bootstrap coefficient realizations according to their empirical risk and prior penalty. The resulting measure is therefore a generalized Bayesian posterior of the Bayesian-bootstrap empirical distribution.
The use of Gibbs updating has previously been implemented for Cox regression~\citep{mcgree2025approach}; however, the approach we develop in this paper differs in that the Gibbs update is applied to a Bayesian-bootstrap distribution of Cox coefficient estimates rather than directly to the Cox parameter space through the partial likelihood.

\subsection{Learning parameter calibration}

The choice of $\lambda$ determines the relative concentration of the empirical risk. The calibration of Gibbs learning parameters has received considerable attention in generalized Bayesian inference~\citep{wu2023comparison}. Existing approaches include calibration based on posterior coverage and predictive performance~\citep{bissiri2016general}, as well as data-driven procedures based on empirical risk and cross-validation~\citep{syring2019calibrating,zhang2006from}.
Because $\widetilde{\ell}^{(m)}$ is an aggregate pairwise loss over $M$ bootstrap draws, its magnitude depends on the number of Monte Carlo samples. We therefore use the flexible scaling $\lambda = c/M^\alpha$, $\alpha\in[0.5,1]$, 
which gives
\begin{eqnarray}
\lambda\widetilde{\ell}^{(m)} = cM^{-\alpha} \sum_{j=1}^{M} L \left(\beta^{(m)},\beta^{(j)}\right).
\label{eq:scaled_empirical_risk}
\end{eqnarray}
This scaling controls the dependence of the Gibbs weights on the Monte Carlo sample size and provides a practical mechanism for stabilizing the empirical-risk contribution. Related considerations concerning learning-rate selection and posterior calibration arise in generalized Bayesian inference~\citep{bissiri2016general,syring2019calibrating,jewson2018principles}.

In the numerical implementation, we use $c=1$ and $\alpha=0.5$, with the resulting learning parameter specified as $\lambda=M^{-1/2}$. This choice is used consistently in the simulation and data analyses. The asymptotic results explained later in this paper are stated separately in terms of a local calibration condition on the \emph{effective Gibbs loss}; this condition should not be interpreted as requiring the particular finite-$M$ implementation $\lambda=M^{-1/2}$ to converge numerically to one.

\subsection{BB-Cox algorithm}

Algorithm~\ref{alg:gb_bbcox} emphasizes the two conceptually distinct components of BB-Cox. First, the Bayesian bootstrap generates an empirical distribution of perturbed Cox coefficient estimates. Second, generalized Bayesian updating redistributes probability mass over these estimates through the empirical risk and prior penalty. The resulting distribution provides a unified basis for point estimation and uncertainty quantification.
\begin{algorithm} 
\caption{Generalized Bayesian-Bootstrap Cox Model (BB-Cox)}
\label{alg:gb_bbcox}
\begin{algorithmic}[1]
\REQUIRE Observations 
$\mathcal{D} = \{(T_i,\delta_i,X_i)\}_{i=1}^n$,
number of bootstrap samples $M$,
learning parameter $\lambda$,
loss function $L(\cdot,\cdot)$,
and prior penalty $\Omega(\cdot)$
\FOR{$m = 1,\dots,M$}
\STATE Draw Dirichlet weights $(w_1^{(m)},\dots,w_n^{(m)}) \sim \text{Dirichlet}(1,\dots,1)$
\STATE Compute weighted Cox estimator $\beta^{(m)} = \arg\min_{\beta} \ell^{(m)}(\beta)$
\ENDFOR
\FOR{$m = 1,\dots,M$}
\STATE Compute empirical risk,  $$\tilde{\ell}^{(m)} = \sum_{j} L\!\left(\beta^{(m)},\beta^{(j)}\right),\quad \text{where, } L\!\left(\beta^{(m)},\beta^{(j)} \right) = \sum_{k=1}^p \left(\beta_k^{(m)}-\beta_k^{(j)}\right)^2$$
\STATE Compute normalized Gibbs weight, $\pi_m \propto \exp\{-\lambda \tilde{\ell}^{(m)} - \Omega(\beta^{(m)})\}$. 
\ENDFOR
\STATE Form the generalized Bayesian posterior of Cox, i.e., BB-Cox,
$${Q}_{\text{BB-Cox}}(\beta) = \sum_{m=1}^M \pi_m \delta_{\beta^{(m)}}(\beta),$$ where $\delta(.)$ is the Dirac measure.
\STATE Compute posterior summaries: generalized posterior mean $\hat{\beta}_{\text{BB-Cox}}$, highest posterior density (HPD) credible intervals from ${Q}_{\text{BB-Cox}}(\beta)$.
\end{algorithmic}
\end{algorithm}
The generalized posterior mean is obtained as
\begin{eqnarray}
\hat{\beta}_{\text{BB-Cox}} = \int \beta\,dQ_{\text{BB-Cox}}(\beta) = \sum_{m=1}^{M} \pi_m\beta^{(m)}.
\label{eq:bbcox_mean}
\end{eqnarray}
Posterior uncertainty can be summarized using highest posterior density (HPD) credible intervals~\citep{box2011bayesian}. For a scalar coefficient $\beta_k$, a $100(1-\alpha)\%$ HPD interval is a shortest interval $[a_k,b_k]$ satisfying
\begin{eqnarray}
Q_{\text{BB-Cox}}\left(a_k\leq\beta_k\leq b_k\right) \geq 1-\alpha, \quad k=1,\ldots,p 
\label{eq:hpd}
\end{eqnarray}
with all points within the interval having posterior density at least as large as points outside the interval. This construction is particularly useful when the resulting generalized posterior is asymmetric or otherwise deviates from a Gaussian approximation.
For interpretation under the Cox model, inference is reported on the hazard-ratio scale: $\text{HR}_{\text{BB-Cox}}=\exp\left(\hat{\beta}_{\text{BB-Cox}}\right)$, with corresponding credible intervals obtained by applying the monotone transformation $\exp(\cdot)$ to the coefficient distribution.

\section{BB-Cox and Bayesian Cox Regression}
\label{sec:BB_Cox_Bayes_relation}

The relationship between BB-Cox and Bayesian Cox regression can be understood by viewing both procedures as probability measures over the regression coefficient $\beta$. Bayesian Cox regression combines the Cox partial likelihood with a prior distribution according to Bayes' theorem~\citep{ibrahim2001bayesian},
\begin{eqnarray}
p(\beta\mid\mathcal{D}) \propto \exp\{\ell_n(\beta)\}\pi(\beta).
\label{eq:bayes_cox}
\end{eqnarray}
Whereas, BB-Cox ($Q_{\text{BB-Cox}}$) constructs the empirical Bayesian-bootstrap measure $Q_{\text{BB}}$ and applies Gibbs reweighting to the bootstrap coefficient realizations. The two procedures therefore differ in construction, but they can share the same first-order asymptotic behavior under appropriate regularity and calibration conditions.
In this section, we will provide some key theorems and related proof where we will show that under certain conditions $Q_{\text{BB-Cox}}(\beta) \approx
Q_{\text{BB}}(\beta)$, $Q_{\text{BB}}(\beta) \approx p(\beta\mid\mathcal{D})$, and $Q_{\text{BB-Cox}}(\beta) \approx p(\beta\mid\mathcal{D})$.

\subsection{Asymptotic assumptions}

We formalize this comparison using the following assumptions.

\begin{assumption}[Cox regularity and local asymptotic normality]
\label{ass:cox_regular}

The Cox proportional hazards model satisfies the standard regularity conditions ensuring consistency and asymptotic normality of the partial-likelihood estimator. In particular,
$$
\sqrt{n}(\hat{\beta}-\beta_0) \Rightarrow N\left(0,I^{-1}(\beta_0)\right).
$$
Moreover, the Cox partial log-likelihood satisfies the local asymptotic normality expansion
$$
\ell_n \left(\hat{\beta}+\frac{h}{\sqrt n}\right)-\ell_n(\hat{\beta})= h^\top Z_n-\frac12h^\top I(\beta_0)h+o_p(1),
$$
uniformly for $h$ in compact subsets of $\mathbb{R}^p$, where 
$$
Z_n \Rightarrow N\left(0,I(\beta_0)\right).
$$
These are standard large-sample conditions for likelihood-based inference in the Cox model~\citep{andersen2012statistical}.
\end{assumption}

\begin{assumption}[Validity of the Bayesian bootstrap]
\label{ass:bb_valid}

The Bayesian-bootstrap weighted Cox estimator admits the same first-order influence-function representation as the ordinary Cox estimator. Consequently,
$$
\mathcal{L}_{Q_{\text{BB}}} \left\{\sqrt n(\beta-\hat{\beta})\right\}\Rightarrow N\left(0,I^{-1}(\beta_0)\right).
$$
\end{assumption}

\begin{assumption}[Locally regular prior]
\label{ass:local_prior}

The prior density satisfies $\pi(\beta)>0$ in a neighbourhood of $\beta_0$, and $\log\pi(\beta)$ is continuously differentiable in that neighbourhood. Moreover, for every finite $C>0$,
$$
\sup_{\|h\|\leq C}\left|\log\pi\left(\hat{\beta}+\frac{h}{\sqrt n}\right) - \log\pi(\hat{\beta}) \right| \rightarrow_p0.
$$
This is the local prior condition under which the prior is asymptotically dominated by the Cox partial likelihood, as in standard Bernstein--von Mises theory~\citep{van1998asymptotic,kleijn2012bernstein}.
\end{assumption}

\begin{assumption}[Local Gibbs calibration]
\label{ass:gibbs_calibration}

The effective Gibbs loss is locally calibrated to the Cox partial likelihood in the sense that, for every finite $C>0$,
$$
\sup_{\|h\|\leq C}\left|\lambda\widetilde{\ell}\left(\hat{\beta}+\frac{h}{\sqrt n}\right) - \mathcal{L}_n \left(\hat{\beta}+\frac{h}{\sqrt n}\right) \right| \rightarrow_p0,
$$
up to an additive term that is constant in $h$ and therefore disappears after normalization of the Gibbs distribution.
This assumption concerns the local geometry induced by the empirical-risk/Gibbs construction. It does not require the finite-$M$ computational choice $\lambda=M^{-1/2}$ to converge to one. Rather, it states the asymptotic calibration required for the effective Gibbs weighting to reproduce the local Cox likelihood geometry.
\end{assumption}

\begin{assumption}[Non-negligible prior regime]
\label{ass:informative_prior}

For an informative prior regime, suppose that
$$
\Omega\left(\hat{\beta}+\frac{h}{\sqrt n}\right)-\Omega(\hat{\beta})\rightarrow \omega(h),
$$
uniformly on compact subsets of $\mathbb{R}^p$, where $\omega(h)$ is finite and non-constant. Thus, the prior contribution remains non-negligible on the local asymptotic scale.
\end{assumption}

\subsection{First-order equivalence of Bayesian bootstrap and Bayesian Cox}

\begin{theorem}[Bayesian bootstrap--Bayesian Cox Bernstein--von Mises equivalence]
\label{thm:bb_bvm}

Under Assumptions~\ref{ass:cox_regular}--\ref{ass:local_prior},
$$
\mathcal{L}_{p}\left\{\sqrt n(\beta-\hat{\beta})\mid\mathcal{D}\right\} \Rightarrow N\left(0,I^{-1}(\beta_0)\right),
$$
$$
\mathcal{L}_{Q_{\text{BB}}}\left\{\sqrt n(\beta-\hat{\beta})\right\} \Rightarrow N\left(0,I^{-1}(\beta_0)\right).
$$
Consequently,
$$
Q_{\text{BB}}(\beta) \approx p(\beta\mid\mathcal{D}),
$$
in the sense of first-order Bernstein--von Mises equivalence.
\end{theorem}

\begin{proof}

By Assumption~\ref{ass:cox_regular}, the Cox partial log-likelihood admits the local expansion
$$
\ell_n\left(\hat{\beta}+\frac{h}{\sqrt n}\right)-\ell_n(\hat{\beta})=h^\top Z_n-\frac12h^\top I(\beta_0)h+o_p(1).
$$
By Assumption~\ref{ass:local_prior}, the local prior contribution satisfies
$$
\log\pi\left(\hat{\beta}+\frac{h}{\sqrt n}\right)-\log\pi(\hat{\beta})=o_p(1)
$$
uniformly for bounded $h$. Hence, the local posterior kernel is asymptotically proportional to
$$
\exp\left\{h^\top Z_n-\frac12h^\top I(\beta_0)h\right\}.
$$
Standard Bernstein--von Mises arguments therefore yield
$$
\sqrt n(\beta-\hat{\beta}) \mid\mathcal D \Rightarrow N\left(0,I^{-1}(\beta_0)\right).
$$
The corresponding Bayesian-bootstrap result follows directly from Assumption~\ref{ass:bb_valid}. Since both local distributions converge to the same Gaussian limit, the Bayesian bootstrap and Bayesian Cox posterior are first-order Bernstein--von Mises equivalent.
\end{proof}

\begin{corollary}[Locally regular Bayesian bootstrap approximation]
\label{cor:bb_bayes}

For a non-informative or weakly informative prior satisfying Assumption~\ref{ass:local_prior},
$$
Q_{\text{BB}}(\beta) \approx p(\beta\mid\mathcal{D}).
$$
\end{corollary}

\subsection{BB-Cox under locally regular priors}

We next establish the corresponding result for BB-Cox.

\begin{theorem}[BB-Cox under a locally regular prior]
\label{thm:bbcox_local}

Suppose Assumptions~\ref{ass:cox_regular}--\ref{ass:gibbs_calibration} hold and the prior satisfies Assumption~\ref{ass:local_prior}. Then
$$
Q_{\text{BB-Cox}}(\beta) \approx Q_{\text{BB}}(\beta),
$$
and consequently
$$
Q_{\text{BB-Cox}}(\beta) \approx p(\beta\mid\mathcal{D}).
$$
The equivalences are understood in the first-order local asymptotic sense.
\end{theorem}

\begin{proof}

Consider the local parameterization
$$
\beta_h = \hat{\beta}+\frac{h}{\sqrt n}.
$$
Under Assumption~\ref{ass:local_prior},
$$
\Omega(\beta_h)-\Omega(\hat{\beta}) = o_p(1)
$$
uniformly for bounded $h$, since $\pi(\beta)\propto\exp\{-\Omega(\beta)\}$.
Thus, the prior term contributes an asymptotically constant multiplicative factor to the local Gibbs weights and disappears upon normalization. By Assumption~\ref{ass:gibbs_calibration}, the effective Gibbs loss has the same local geometry as the Cox negative partial log-likelihood, up to terms that are asymptotically negligible or constant in $h$. Therefore, the Gibbs reweighting does not change the first-order local distribution generated by the Bayesian-bootstrap coefficient estimates. Hence,
$$
Q_{\text{BB-Cox}} \approx Q_{\text{BB}}.
$$
The result then follows by Theorem~\ref{thm:bb_bvm}.
\end{proof}
\noindent
Thus, for a locally regular prior,
\begin{eqnarray}
Q_{\text{BB-Cox}}(\beta) \approx Q_{\text{BB}}(\beta) \approx p(\beta\mid\mathcal{D}).
\label{eq:threeway_local}
\end{eqnarray}
This is an asymptotic statement and these three procedures remain distinct at finite samples because they are constructed from different probability measures and different updating mechanisms.

\subsection{BB-Cox with a non-negligible prior}

The preceding equivalence changes when prior information remains non-negligible on the local asymptotic scale. In this setting, the Bayesian bootstrap alone does not contain the prior contribution and therefore need not reproduce the Bayesian Cox posterior.
Under Assumption~\ref{ass:informative_prior},
$$
\Omega
\left(\hat{\beta}+\frac{h}{\sqrt n}\right)-\Omega(\hat{\beta})\rightarrow\omega(h),
$$
where $\omega(h)$ is non-constant. Consequently, the prior modifies the local posterior kernel. Combining the LAN expansion with the prior contribution gives the local Bayesian Cox kernel
$$
\exp\left\{h^\top Z_n-\frac12h^\top I(\beta_0)h-\omega(h)\right\}.
$$
Unlike the locally regular case, this limiting distribution need not be Gaussian because $\omega(h)$ may alter the local shape of the posterior.

\begin{theorem}[BB-Cox with a non-negligible prior]
\label{thm:bbcox_info}

Suppose Assumptions~\ref{ass:cox_regular}, \ref{ass:bb_valid}, \ref{ass:gibbs_calibration}, and~\ref{ass:informative_prior} hold, and suppose that the penalty $\Omega(\beta)$ used in BB-Cox is the same penalty satisfying
$$
\pi(\beta)\propto\exp\{-\Omega(\beta)\}
$$
in the Bayesian Cox model. Then BB-Cox and Bayesian Cox have the same first-order local posterior kernel:
$$
Q_{\text{BB-Cox}}(\beta) \approx p(\beta\mid\mathcal{D}).
$$
Whereas,
$$
Q_{\text{BB}}(\beta) \not\approx p(\beta\mid\mathcal{D})
$$
in general, because $Q_{\text{BB}}$ does not incorporate the non-negligible prior contribution.
\end{theorem}

\begin{proof}

Under Assumption~\ref{ass:cox_regular}, the Cox likelihood contribution on the local scale is
$$
\exp\left\{h^\top Z_n-\frac12h^\top I(\beta_0)h\right\}.
$$
Under Assumption~\ref{ass:informative_prior}, the prior contributes the non-constant factor
$$
\exp\{-\omega(h)\}.
$$
Therefore, the local Bayesian Cox posterior is proportional to
$$
\exp\left\{h^\top Z_n-\frac12h^\top I(\beta_0)h-\omega(h)\right\}.
$$
By Assumption~\ref{ass:gibbs_calibration}, the effective Gibbs loss used by BB-Cox reproduces the local Cox likelihood geometry. Because BB-Cox uses the same penalty $\Omega$, its local Gibbs kernel contains the same prior contribution $\exp\{-\omega(h)\}$. Consequently, BB-Cox and Bayesian Cox have the same first-order local posterior kernel.
The Bayesian bootstrap distribution, however, is generated solely by Dirichlet perturbations of the empirical distribution and contains no corresponding prior penalty. Therefore, when $\omega(h)$ is non-constant, $Q_{\text{BB}}$ does not generally have the same local limit as the Bayesian Cox posterior.
\end{proof}

The theorem clarifies an important distinction between the three procedures. The Bayesian bootstrap captures uncertainty arising from empirical-distribution perturbations, whereas BB-Cox retains this bootstrap uncertainty while allowing the Gibbs step to introduce prior information. Bayesian Cox regression incorporates the same prior information directly through Bayes' theorem. Thus, when the prior has a non-negligible local contribution, BB-Cox can reproduce the first-order prior-modified posterior geometry provided that the Gibbs loss is appropriately calibrated and the same prior penalty is used.

\section{Simulation study}\label{sec:sim}

In this section, we conduct a simulation study to compare the estimates obtained from the BB-Cox model with those from the frequentist and Bayesian Cox models. We further investigate the impact of different prior specifications on the BB-Cox estimates, and check the asymptotic equivalence of BB-Cox and Bayesian Cox models.

\subsection{Simulation design}

For the BB-Cox model, the data were generated from a Cox proportional hazards model with a binary treatment covariate. The treatment parameter was fixed at $\beta_{\text{true}} = 0.7$, representing a moderate treatment effect. Independent right-censoring times were generated to reflect a realistic survival setting. 
Each simulated dataset consisted of sample sizes $n \in \{50, 100, 150, 200\}$, representing small to moderately large studies. We use Gibbs posterior tuning parameter $\lambda=1/\sqrt{M}$, where $M$ is the number of bootstrap samples.
The simulation framework also considered three prior configurations: (i) scenario - 1: a non-informative prior, representing minimal prior information using $N(0,10^4)$; (ii) scenario - 2: a weakly informative prior as $N(0,3)$; and (iii) scenario - 3: an informative prior $N(0.7,0.5)$, introducing stronger prior influence. These priors were designed to assess the sensitivity of the proposed estimator to prior misspecification and varying degrees of prior strength.
For each combination of $n$ and scenarios, survival datasets were generated independently, and the BB-Cox, frequentist and Bayesian Cox models were applied. This process has been replicated for 1,000 times. Performance was evaluated using the root mean squared error (RMSE) and mean absolute error (MAE) matrices of the estimated regression coefficient $\hat{\beta}$ relative to the true value $\beta_{\text{true}} = 0.7$. For Bayesian cox, we used 20,000 Markov chain Monte Carlo (MCMC) samples to fit the model with 2 chains and first 10,000 observations were used as burn-in. We then derived the posterior summaries from the Bayesian Cox model and compared with the proposed BB-Cox approach for different prior scenarios. 

\subsection{Simulation results}

Figure~\ref{fig:bbcox_sim_valstat} presents a comparison of estimation accuracy for the Cox regression coefficient $\beta = 0.7$. We evaluate the accuracy using the root mean squared error (RMSE) on the left panels and the mean absolute error (MAE) on the right panels. 
Across all scenarios and methods, a clear and consistent pattern emerges: both RMSE and MAE decrease as the sample size increases. This is evident from the downward trend of all curves. For example, in Scenario 1, RMSE decreases from approximately 0.45 at a sample size of 50 to about 0.19 at a sample size of 200, and MAE shows a similar decline from around 0.36 to 0.15. This pattern demonstrates that all three estimators are consistent.

In scenario 1, i.e., for non-informative prior (top row of Figure~\ref{fig:bbcox_sim_valstat}), the three methods perform almost identically. Any differences are very small, with Bayesian Cox and BB-Cox showing only a marginal improvement over the classical Cox model at the smallest sample size, suggesting that when conditions are relatively well-behaved, all three approaches yield comparable estimates of the true value of $\beta = 0.7$. In scenario 2, the results are again very similar across methods, but with slightly more noticeable differences at smaller sample sizes. At a sample size of 50, the Bayesian Cox and BB-Cox methods exhibit slightly lower RMSE (around 0.35) compared to the Cox model (slightly higher). The same pattern appears in the MAE panel. However, as the sample size increases to 150 and 200, these differences become negligible. Scenario 3 that represents  informative prior, highlights the most noticeable differences between the methods. At the smallest sample size ($n = 50$), the Bayesian Cox and BB-Cox approaches perform better than the frequentist Cox model. This indicates that, under the informative prior scenario, the classical Cox estimator has higher variability when data are limited. However, as the sample size increases, the model improves rapidly. 
Comparing RMSE and MAE across all panels of Figure~\ref{fig:bbcox_sim_valstat} shows consistent patterns. RMSE, which penalizes larger errors more heavily, and MAE, which reflects average absolute deviation, both lead to the same ranking of methods in each scenario. This consistency indicates that the observed differences are not driven by a few extreme outliers but rather reflect systematic differences in estimation performance.
\begin{figure}[!htp]
    \centering
    \includegraphics[width=0.98\textwidth]{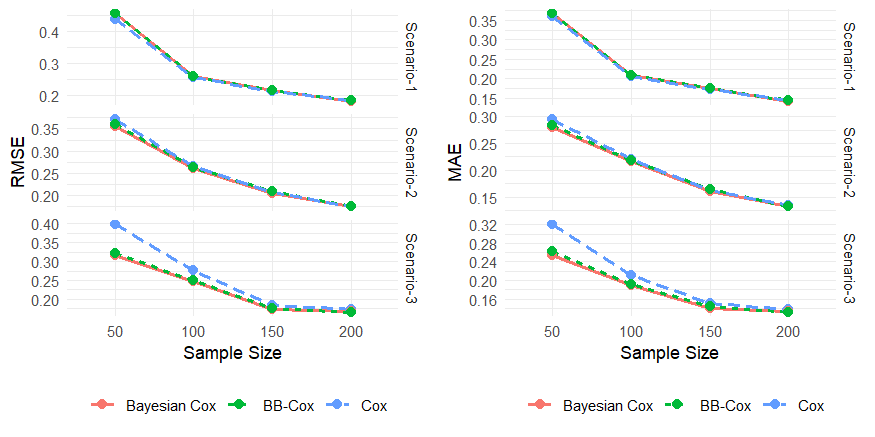}
    \caption{Comparison between Cox and BB-Cox Gibbs posterior of coefficient $\beta$ using root mean squared error (RMSE) performance metric for Gibbs tuning parameters $\lambda=0.1,1$ and $2$ three different prior survival curves scenarios: (1) non-informative $N(0,10^4)$, (2) weakly-informative $N(0,3)$ and (3) informative priors $N(0.7,0.5)$.}
    \label{fig:bbcox_sim_valstat}
\end{figure}

We also evaluated the empirical coverage probabilities of the nominal 95\% intervals under the three prior specifications (see Table~\ref{tab:coverage}). Across all scenarios and sample sizes, BB-Cox and Bayesian Cox show nearly indistinguishable coverage, with differences generally no larger than one percentage point. Coverage remains close to the nominal level throughout, typically ranging from 0.88 to 0.96. These results indicate that BB-Cox provides reliable uncertainty quantification, closely reproducing the coverage properties of Bayesian Cox inference while maintaining satisfactory frequentist performance across the settings considered. We have also seen that the computation time for both Cox and BB-Cox is less than one minute, whereas for Bayesian Cox using Stan programming language it was on average 3 minutes using Windows OS (Intel(R) Core(TM) Ultra 7 165U (1.70 GHz), 32.0 GB RAM).

\begin{table}[!htp]
\centering
\caption{Coverage probabilities across methods, scenarios, and sample sizes. Average computation time is also provided.}\label{tab:coverage}
\centering
\begin{tabular}[t]{lrrrr}\toprule
Scenario & Sample Size & BB-Cox & Bayesian Cox & Cox\\
\midrule
Scenario-1 & 50 & 0.92 & 0.94 & 0.94\\
Scenario-1 & 100 & 0.92 & 0.94 & 0.95\\
Scenario-1 & 150 & 0.93 & 0.95 & 0.95\\
Scenario-1 & 200 & 0.93 & 0.95 & 0.95\\
\addlinespace
Scenario-2 & 50 & 0.94 & 0.95 & 0.95\\
Scenario-2 & 100 & 0.93 & 0.94 & 0.95\\
Scenario-2 & 150 & 0.92 & 0.95 & 0.95\\
Scenario-2 & 200 & 0.92 & 0.93 & 0.94\\
\addlinespace
Scenario-3 & 50 & 0.95 & 0.97 & 0.96\\
Scenario-3 & 100 & 0.95 & 0.96 & 0.95\\
Scenario-3 & 150 & 0.95 & 0.96 & 0.96\\
Scenario-3 & 200 & 0.95 & 0.96 & 0.96\\
\midrule
\multicolumn{2}{l}{Computation time (average)} &  $<1$ min& 5-min &  $<1$ min\\
\bottomrule
\end{tabular}
\end{table}

We further explore the distributions of the regression coefficient estimates from the Bayesian Cox and BB-Cox methods across $R=1,000$ simulated datasets in Figure~\ref{fig:bbcox_sim_estimates}. For BB-Cox, each analysis is based on $M=1,000$ Bayesian bootstrap replicates. 
We see that both estimators are approximately unbiased, with sampling distributions centered near the true parameter value. We can also see that the estimation variability decreases with sample size.
Across all sample sizes, BB-Cox and Bayesian Cox display nearly identical centers and spreads. Thus, the BB-Cox reproduces the principal finite-sample efficiency gains of Bayesian Cox regression while retaining a bootstrap-based construction, which is consistent with the asymptotic relationship established in Section~\ref{sec:BB_Cox_Bayes_relation}.

\begin{figure}[!htp]
    \centering
    \includegraphics[width=0.96\textwidth]{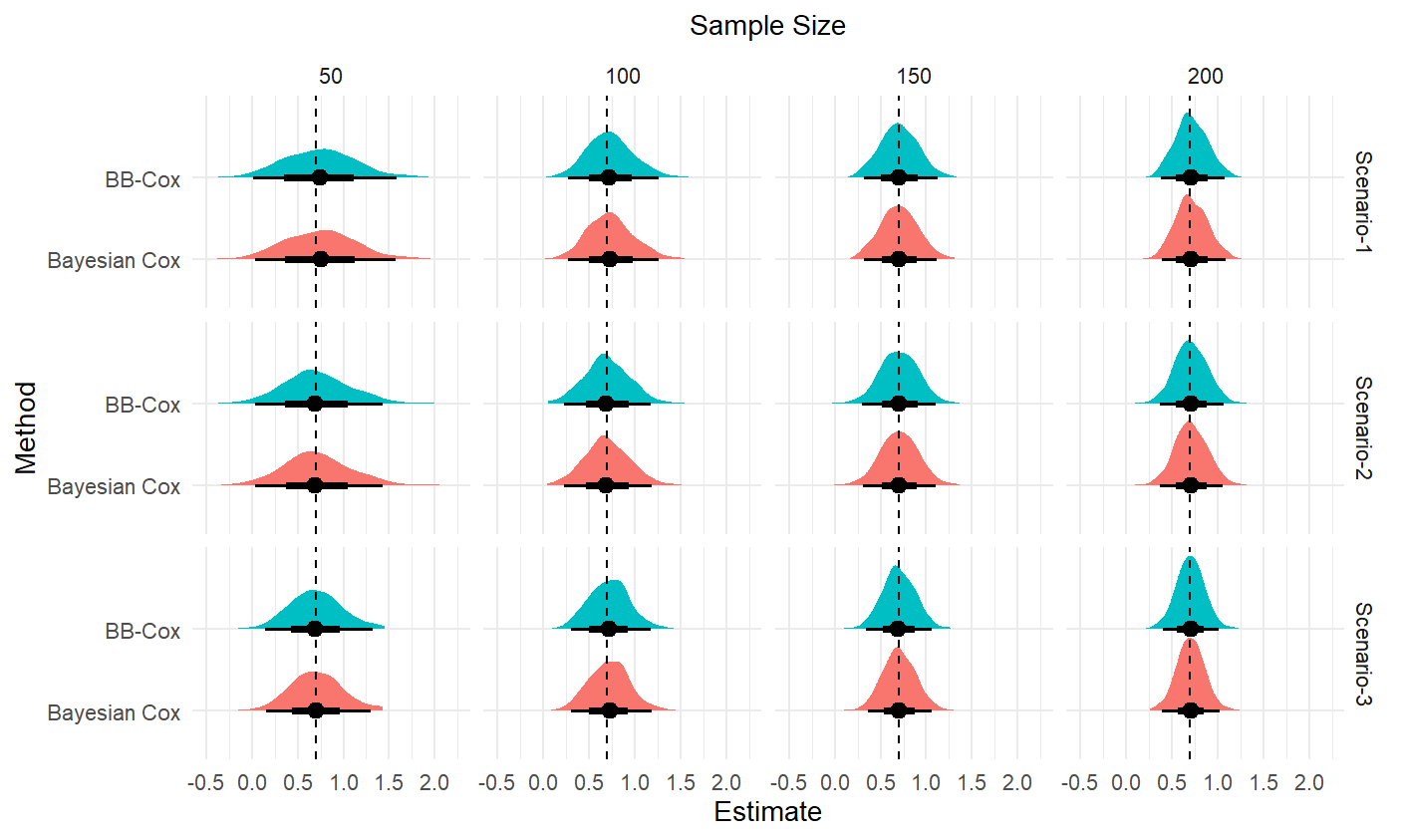}
    \caption{Distribution of $\beta$ from BB-Cox and Bayesian Cox models under three prior scenarios: (1) non-informative $N(0,10^4)$, (2) weakly informative $N(0,3)$, and (3) informative $N(0.7,0.5)$, across sample sizes $n=50,100,150,$ and $200$.}
    \label{fig:bbcox_sim_estimates}
\end{figure}


\section{Example}\label{sec:ex}

We illustrate the proposed method using kidney transplantation data from the Australia and New Zealand Dialysis and Transplant (ANZDATA) Registry. This application demonstrates the utility of the proposed approach in a registry-based survival analysis setting with clinically relevant donor risk stratification and right-censored time-to-event outcomes. The cohort includes adult recipients undergoing their first kidney transplant between January 2005 and December 2019, with follow-up available until December 2020. 
The outcome is time to 3-year all-cause allograft loss, defined as graft failure from any cause within three years after transplantation. Recipients without an event are right censored at three years, last follow-up, or study end, whichever occurs first.

The exposure of interest is donor kidney quality measured using the kidney donor risk index (KDRI), a validated measure of donor risk based on donor characteristics and their association with graft failure~\citep{rao2009comprehensive}. Following recent applications of donor risk stratification~\citep{zhu2025survival}, KDRI is dichotomised using the 90th percentile as the cut-off, with donor kidneys at or above this threshold classified as high-KDRI and those below classified as low-KDRI.
The association between KDRI category and allograft loss is estimated using the proposed adjusted BB-Cox model. The model includes KDRI category as the primary exposure and adjusts for recipient gender, coronary artery disease, peripheral vascular disease, cerebrovascular disease, diabetes status, donor-specific human leukocyte antigen (HLA) antibody mismatches and ethnicity group~\citep{prunster2023kidney}. The parameter of interest is the hazard ratio comparing recipients receiving high-KDRI versus low-KDRI donor kidneys after adjustment for these covariates. 
Only patients with a first graft loss were included, yielding a cohort of 15,454 participants. After excluding those with missing KDRI data, the final study cohort consisted of 9,463 patients, who were included in the final analysis. Over the study period, 961 (10.15\%) recipients lost their allografts within 3 years post-transplant. 

\subsection{Results}

Table~\ref{tab:results} presents the estimated hazard ratios (HR) for the proposed BB-Cox, the Bayesian Cox and frequentist Cox models. Here we used a non-informative prior $N(0,10^2)$ for the model parameters for the BB-Cox and Bayesian Cox models. As the primary exposure variable, KDRI demonstrated a consistent effect with the outcome across all three models. Compared with recipients of kidneys with low KDRI, those receiving high-KDRI kidneys had approximately a twofold higher hazard of allograft loss within three years (BB-Cox: HR = 2.08, 95\% BCI: 1.57-2.46; Bayesian Cox: HR = 2.09, 95\% BCI: 1.64-2.59; frequentist Cox: HR = 2.10, 95\% CI: 1.77-2.49). The estimated HR and corresponding uncertainty intervals were similar across the three modeling approaches, indicating that the proposed BB-Cox model produces results comparable to those obtained from the frequentist Cox and conventional Bayesian models under non-informative prior, while requiring substantially less computation than the Bayesian Cox model.

\begin{table}[!htp]
\centering
\caption{Model based estimates of hazard ratio (HR) from BB-Cox, Bayesian Cox and frequentist Cox models with outcome variable as
patients lost their allografts within 3 years post-transplant. Here, BCI refers to the Bayesian highest posterior density credible interval and CI as frequentist confidence interval. Non-informative prior $N(0,10^2)$ is considered for the BB-Cox and Bayesian Cox model estimation. First category of each variable is considered as the reference category.}\label{tab:results}
\centering
\begin{tabular}[t]{llccc} \toprule
\multicolumn{2}{c}{Variables} & BB-Cox & Bayesian Cox & Cox\\ 
Characteristics & Categories  & HR (95\% BCI) & HR (95\% BCI) & HR (95\% CI) \\ \midrule
KDRI & Low & - & - & - \\ 
     & High & 2.08 (1.57, 2.46) & 2.09 (1.64, 2.59) & 2.10 (1.77, 2.49)\\ \addlinespace
Gender & Female & - & - & - \\ 
 & Male & 1.09 (0.94, 1.29) & 1.09 (0.95, 1.26) & 1.09 (0.95, 1.26)\\ \addlinespace
Coronary artery & No & - & - & - \\  
disease & Yes & 1.20 (0.94, 1.45) & 1.20 (0.98, 1.43) & 1.20 (1.00, 1.44)\\ \addlinespace
Peripheral vascular & No & - & - & - \\ 
disease & Yes & 1.12 (0.80, 1.49) & 1.13 (0.78, 1.52) & 1.12 (0.89, 1.42)\\ \addlinespace
Cerebrovascular & No & - & - & -\\
disease  & Yes & 1.44 (1.00, 1.99) & 1.43 (1.00, 1.98) & 1.44 (1.09, 1.90)\\ \addlinespace
Diabetes & No & - & - & - \\ 
  & Yes & 1.19 (1.00, 1.39) & 1.18 (1.00, 1.40) & 1.18 (1.02, 1.38)\\ \addlinespace
HLA-DR & 0 & - & - & - \\ 
mismatches & 1 & 1.33 (1.07, 1.65) & 1.32 (1.08, 1.69) & 1.32 (1.10, 1.58)\\ 
 & 2 & 1.40 (1.13, 1.71) & 1.40 (1.12, 1.77) & 1.39 (1.17, 1.66)\\ \addlinespace
Ethnicity & Caucasian & - & - & - \\ 
 & Asian & 0.82 (0.66, 1.02) & 0.82 (0.65, 1.01) & 0.83 (0.67, 1.02)\\ 
 & Indigenous & 1.41 (1.00, 1.86) & 1.40 (1.00, 1.79) & 1.43 (1.10, 1.85)\\ 
 & Maori & 1.12 (0.62, 1.85) & 1.12 (0.58, 1.85) & 1.16 (0.73, 1.84)\\ 
 & Others & 0.82 (0.68, 0.99) & 0.80 (0.67, 0.98) & 0.82 (0.69, 0.98)\\ \midrule
\multicolumn{2}{l}{Computation time} & $<1$ min & $>5$ min & $<1$ min \\ 
\bottomrule
\end{tabular}
\end{table}


Sometimes, researchers are interested in examining whether the effect of the KDRI persists among patients whose kidney allografts survive beyond a certain period after transplantation. To investigate this, separate analyses were conducted by restricting the study population to patients whose allografts remained functional for more than 5, 10, and 15 years post-transplant. Table~\ref{tab:results2} presents the HR estimates from the BB-Cox model for these three survival thresholds. Across all analyses, recipients of kidneys with a high KDRI consistently experienced a significantly higher risk of allograft loss than those receiving kidneys with a low KDRI. Specifically, the estimated HRs were 1.9 (95\% BCI: 1.54-2.34), 1.79 (95\% BCI: 1.49-2.10), and 1.81 (95\% BCI: 1.49-2.12) for patients with graft survival exceeding 5, 10, and 15 years, respectively. The similarity of these estimates suggests that the adverse effect of a high KDRI on long-term allograft failure remains relatively stable, even among recipients whose grafts have already survived for many years after transplantation.

\begin{table}[!htp]
\centering
\caption{BB-Cox based HR estimates for patients with allograft lost more than 5, 10 and 15 years of post-transplant. Non-informative prior $N(0,10^2)$ is considered for the BB-Cox model. Here, BCI refers to the Bayesian highest posterior density credible interval, and the first category of each variable is considered as the reference category.}\label{tab:results2}
\centering
\begin{tabular}[t]{llccc} \toprule
\multicolumn{2}{c}{Variables} & 5-year & 10-year & 15-year \\ 
Characteristics & Categories  & HR (95\% BCI) & HR (95\% BCI) & HR (95\% CI) \\ \midrule
KDRI & Low & - & - & - \\ 
     & High & 1.90 (1.54, 2.34) & 1.79 (1.49, 2.10) & 1.81 (1.49, 2.12)\\ \addlinespace
Gender & Female & - & - & - \\ 
 & Male & 1.10 (0.94, 1.29) & 1.09 (0.96, 1.23) & 1.11 (0.96, 1.24)\\ \addlinespace
Coronary artery & No & - & - & - \\  
disease & Yes & 1.24 (0.94, 1.51) & 1.27 (1.08, 1.50) & 1.36 (1.13, 1.58)\\ \addlinespace
Peripheral vascular & No & - & - & - \\ 
disease & Yes & 1.18 (0.90, 1.57) & 1.21 (0.97, 1.48) & 1.17 (0.94, 1.45)\\ \addlinespace
Cerebrovascular & No & - & - & -\\
disease  & Yes & 1.47 (0.99, 2.00) & 1.35 (0.95, 1.70) & 1.42 (1.06, 1.77)\\ \addlinespace
Diabetes & No & - & - & - \\ 
  & Yes & 1.15 (0.97, 1.38) & 1.17 (1.01, 1.34) & 1.15 (1.00, 1.31)\\ \addlinespace
HLA-DR & 0 & - & - & - \\ 
mismatches & 1 & 1.29 (1.00, 1.59) & 1.22 (1.04, 1.41) & 1.20 (1.01, 1.40)\\ 
 & 2 & 1.37 (1.10, 1.68) & 1.29 (1.06, 1.49) & 1.27 (1.09, 1.48)\\ \addlinespace
Ethnicity & Caucasian & - & - & - \\ 
 & Asian & 0.79 (0.62, 1.04) & 0.85 (0.71, 1.05) & 0.86 (0.71, 1.02)\\ 
 & Indigenous & 1.59 (1.18, 2.12) & 1.63 (1.23, 2.09) & 1.60 (1.21, 2.10)\\ 
 & Maori & 1.19 (0.59, 1.86) & 1.17 (0.69, 1.75) & 1.22 (0.78, 1.74)\\ 
 & Others & 0.83 (0.66, 1.02) & 0.97 (0.84, 1.13) & 0.98 (0.83, 1.13)\\ 
\bottomrule
\end{tabular}
\end{table}

\subsection{Prior sensitivity}

To evaluate the sensitivity of the proposed BB-Cox framework to different prior assumptions, we consider prior configurations for the Cox regression coefficient vector $\beta \in \mathbb{R}^{p}$. We use the prior penalty function $\Omega(\beta)$ in the generalized Bayesian updating step with different shrinkage structures, such as the Horseshoe and Cauchy priors. These alternative prior specifications allow us to investigate the influence of different shrinkage assumptions on the estimation of $\beta$ and assess the robustness of the BB-Cox model under varying degrees of coefficient regularization.

Let us now define the Horseshoe and Cauchy priors under the BB-Cox model. 
For the Horseshoe prior, the hierarchical formulation is written as
\begin{eqnarray}
\beta_j \mid \lambda_j,\tau &\sim& N(0,\tau^2\lambda_j^2), \quad \lambda_j \sim C^+(0,1), \quad \tau \sim C^+(0,1),
\end{eqnarray}
where $C^+(0,1)$ denotes the standard half-Cauchy distribution, $\lambda_j$ is the local shrinkage parameter, and $\tau$ is the global shrinkage parameter. In this paper, when optimization is performed, we implement the commonly used approximation to the marginal Horseshoe penalty for sparse signals~\citep{carvalho2010horseshoe} as:
\begin{eqnarray}
\Omega_{\text{HS}}(\beta)&\propto&\sum_{j=1}^{p}\log\left(1+\frac{2\tau^2}{\beta_j^2}\right).
\end{eqnarray}
Now, for defining the Cauchy prior, $\beta_j \sim C(0,s)$, we write the density $p_{\text{C}}$ and corresponding penalty function ignoring constant terms as:
\begin{eqnarray}
p_{\text{C}}(\beta_j) = \frac{1}{\pi s\left(1+\beta_j^2/s^2\right)}; \quad \Omega_{\text{C}}(\beta) \propto \sum_{j=1}^{p}\log\left(1+\frac{\beta_j^2}{s^2}\right).
\end{eqnarray}

\begin{figure}[!htp]
    \centering
    \includegraphics[width=0.6\textwidth]{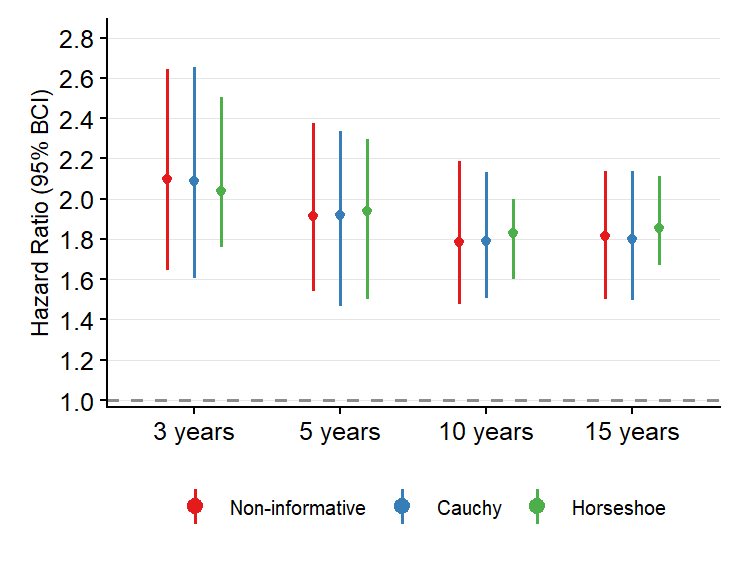}
    \caption{Estimated HRs from the BB-Cox model under non-informative, Cauchy and Horseshoe priors for allograft loss beyond 3, 5, 10, and 15 years post-transplant. Here, BCI refers to the Bayesian credible interval.}
    \label{fig:prior_sen}
\end{figure}

Figure~\ref{fig:prior_sen} presents the HR estimates for KDRI under different prior specifications and follow-up periods. We set $\tau=1$ for the Horseshoe prior and $s=2.5$ for the Cauchy prior. The results are generally consistent across priors, indicating that the association between KDRI and allograft loss is robust to prior assumptions. As expected, the Horseshoe prior provides stronger shrinkage compared with the non-informative and Cauchy priors, leading to slightly more conservative estimates. The HR estimates gradually decrease with longer post-transplant survival thresholds, suggesting a modest attenuation of the KDRI effect among patients with longer-lasting grafts. However, the association remains evident even beyond 10 and 15 years post-transplant.

\section{Conclusion and future work}\label{sec:con}

In this paper, we introduce a Bayesian bootstrap-induces framework for survival analysis that combines the flexibility of the Bayesian bootstrap with the ability to incorporate prior information through generalized Bayesian (Gibbs) updating. Unlike conventional Bayesian survival models, the proposed approach does not require specification of a full probabilistic likelihood and avoids potential restrictive distributional assumptions while maintaining a reasonable Bayesian interpretation. By identifying uncertainty obtained through Bayesian bootstrap perturbations and prior information incorporation via loss-based reweighting, the framework provides a computationally efficient and robust alternative for inference in survival analysis. 
The propose method can be applied to a broad class of survival estimators and as a concrete illustration, in this paper, we developed it for the Cox proportional hazard model. This approach preserves the familiar interpretation of hazard ratios while enabling the incorporation of prior knowledge in a likelihood-free manner. To facilitate practical application, the methodology has been implemented in the open-source \texttt{R} package \texttt{BayesBoots}, downloadable from \href{https://github.com/ksbakar/BayesBoots}{GitHub}. 

Several directions for future research remain. Although this paper focuses on Cox proportional hazards regression, the framework naturally extends to a broad range of survival models, including accelerated failure time models and frailty models~\citep{andersen2012statistical} or joint models for survival analysis~\citep{bakar2024dynamic,guo2026semiparametric}. 
Beyond regression parameters, future research could consider Bayesian bootstrap-induced inference for cumulative incidence functions~\citep{fine1999proportional} and restricted mean survival time~\citep{zhou2023transformation}, which have received increasing attention as clinically interpretable alternatives to hazard-based measures.
Another important direction of the proposed framework relates to the investigation on the choice of the loss function~\citep{bartlett2006convexity}, hierarchical parameters of the prior distribution and posterior predictive performance~\citep{gelman2013bayesian}. Furthermore, it is also useful to explore this approach on intractable likelihood~\citep{matsubara2024generalized} related to survival models. 
From a computational perspective, the proposed framework offers opportunities for inference in high-dimensional and large-scale survival settings. Future research could explore efficient algorithms, adaptive learning-rate selection~\citep{lyddon2019general}, sparse prior formulations for variable selection, and extensions to distributed or online computation~\citep{emek2011online, angelopoulos2024online}. These developments would further broaden the applicability of Bayesian bootstrap-induced generalized Bayesian inference as a flexible and practical alternative to conventional Bayesian survival analysis.

\section*{Acknowledgements}\label{ack}

The authors acknowledge the Centre for Kidney Research (CKR), Kids Research Institute, The Children's Hospital at Westmead, Sydney, New South Wales, Australia, for their support and for facilitating access to the ANZDATA registry. 

\bibliographystyle{chicago}
\bibliography{BayesianBoot_ref}

@book{andersen2012statistical,
  title={Statistical models based on counting processes},
  author={Andersen, Per K and Borgan, Ornulf and Gill, Richard D and Keiding, Niels},
  year={2012},
  publisher={Springer Science \& Business Media}
}

@article{angelopoulos2024online,
  title={Online computation with untrusted advice},
  author={Angelopoulos, Spyros and D{\"u}rr, Christoph and Jin, Shendan and Kamali, Shahin and Renault, Marc},
  journal={Journal of Computer and System Sciences},
  volume={144},
  pages={103545},
  year={2024},
  publisher={Elsevier}
}

@article{bakar2024dynamic,
  title={Dynamic prediction of kidney allograft and patient survival using post-transplant estimated glomerular filtration rate trajectory},
  author={Bakar, Khandoker Shuvo and Teixeira-Pinto, Armando and Gately, Ryan and Boroumand, Farzaneh and Lim, Wai H and Wong, Germaine},
  journal={Clinical Kidney Journal},
  volume={17},
  number={11},
  pages={sfae314},
  year={2024},
  publisher={Oxford University Press}
}

@article{ballante2025generalized,
  title={A generalized Bayesian ensemble survival tree (GBEST) model},
  author={Ballante, Elena and Muliere, Pietro and Figini, Silvia},
  journal={Statistics},
  volume={59},
  number={6},
  pages={1521--1536},
  year={2025},
  publisher={Taylor \& Francis}
}

@article{blei2017variational,
  title={Variational inference: A review for statisticians},
  author={Blei, David M and Kucukelbir, Alp and McAuliffe, Jon D},
  journal={Journal of the American statistical Association},
  volume={112},
  number={518},
  pages={859--877},
  year={2017},
  publisher={Taylor \& Francis}
}

@article{breslow1974covariance,
  title={Covariance analysis of censored survival data},
  author={Breslow, Norman},
  journal={Biometrics},
  volume={30},
  number={1},
  pages={89--99},
  year={1974},
  publisher={JSTOR}
}

@article{bartlett2006convexity,
  title={Convexity, classification, and risk bounds},
  author={Bartlett, Peter L and Jordan, Michael I and McAuliffe, Jon D},
  journal={Journal of the American Statistical Association},
  volume={101},
  number={473},
  pages={138--156},
  year={2006},
  publisher={Taylor \& Francis}
}

@article{bartovs2022informed,
  title={Informed Bayesian survival analysis},
  author={Barto{\v{s}}, Franti{\v{s}}ek and Aust, Frederik and Haaf, Julia M},
  journal={BMC Medical Research Methodology},
  volume={22},
  number={1},
  pages={238},
  year={2022},
  publisher={Springer}
}

@article{bissiri2016general,
  title={A general framework for updating belief distributions},
  author={Bissiri, Pier Giovanni and Holmes, Christopher C and Walker, Stephen G},
  journal={Journal of the Royal Statistical Society: Series B (Statistical Methodology)},
  volume={78},
  number={5},
  pages={1103--1130},
  year={2016},
  publisher={Wiley}
}

@book{box2011bayesian,
  title={Bayesian inference in statistical analysis},
  author={Box, George EP and Tiao, George C},
  year={2011},
  publisher={John Wiley \& Sons}
}

@article{carvalho2010horseshoe,
  title={The horseshoe estimator for sparse signals},
  author={Carvalho, Carlos M and Polson, Nicholas G and Scott, James G},
  journal={Biometrika},
  volume={97},
  number={2},
  pages={465--480},
  year={2010},
  publisher={Oxford University Press}
}

@book{chen2012monte,
  title={Monte Carlo methods in Bayesian computation},
  author={Chen, Ming-Hui and Shao, Qi-Man and Ibrahim, Joseph G},
  year={2012},
  publisher={Springer Science \& Business Media}
}

@article{cox1972regression,
  title={Regression models and life-tables},
  author={Cox, David R},
  journal={Journal of the royal statistical society: Series B (methodological)},
  volume={34},
  number={2},
  pages={187--202},
  year={1972},
  publisher={Wiley Online Library}
}

@article{efron1977efficiency,
  title={The efficiency of Cox's likelihood function for censored data},
  author={Efron, Bradley},
  journal={Journal of the American statistical Association},
  volume={72},
  number={359},
  pages={557--565},
  year={1977},
  publisher={Taylor \& Francis}
}

@article{efron1979bootstrap,
  title={Bootstrap methods: another look at the jackknife},
  author={Efron, Bradley},
  journal={Annals of Statistics},
  volume={7},
  number={1},
  pages={1--26},
  year={1979},
  publisher={Springer}
}

@article{emek2011online,
  title={Online computation with advice},
  author={Emek, Yuval and Fraigniaud, Pierre and Korman, Amos and Ros{\'e}n, Adi},
  journal={Theoretical Computer Science},
  volume={412},
  number={24},
  pages={2642--2656},
  year={2011},
  publisher={Elsevier}
}

@article{fine1999proportional,
  title={A proportional hazards model for the subdistribution of a competing risk},
  author={Fine, Jason P and Gray, Robert J},
  journal={Journal of the American statistical association},
  volume={94},
  number={446},
  pages={496--509},
  year={1999},
  publisher={Taylor \& Francis}
}

@book{fleming2013counting,
  title={Counting processes and survival analysis},
  author={Fleming, Thomas R and Harrington, David P},
  year={2013},
  publisher={John Wiley \& Sons}
}

@article{gamerman1991dynamic,
  title={Dynamic Bayesian models for survival data},
  author={Gamerman, Dani},
  journal={Journal of the Royal Statistical Society Series C: Applied Statistics},
  volume={40},
  number={1},
  pages={63--79},
  year={1991},
  publisher={Oxford University Press}
}

@book{gelman2013bayesian,
  title     = {Bayesian Data Analysis},
  author    = {Gelman, Andrew and Carlin, John B. and Stern, Hal S. and Dunson, David B. and Vehtari, Aki and Rubin, Donald B.},
  edition   = {3},
  year      = {2013},
  publisher = {Chapman and Hall/CRC}
}

@article{guo2026semiparametric,
  title={Semiparametric Joint Modeling for Survival Analysis with Longitudinal Covariates},
  author={Guo, Wensheng and Wang, Tianhao},
  journal={Journal of the American Statistical Association},
  volume = {0},
  number = {0},
  pages={1--13},
  year={2026},
  publisher={Taylor \& Francis}
}

@article{hanson2006modeling,
  title={Modeling censored lifetime data using a mixture of gammas baseline},
  author={Hanson, Timothy E},
  journal = {Bayesian Analysis},
  volume  = {1},
  number  = {3},
  pages   = {575--594},
  year={2006}
}

@article{hjort1990nonparametric,
  title={Nonparametric Bayes estimators based on beta processes in models for life history data},
  author={Hjort, Nils Lid},
  journal={the Annals of Statistics},
  volume  = {18},
  number  = {3},
  pages={1259--1294},
  year={1990},
  publisher={JSTOR}
}

@book{hosmer2008applied,
  title={Applied survival analysis: regression modeling of time-to-event data},
  author={Hosmer Jr, David W and Lemeshow, Stanley and May, Susanne},
  year={2008},
  publisher={John Wiley \& Sons}
}

@book{ibrahim2001bayesian,
  title     = {Bayesian Survival Analysis},
  author    = {Ibrahim, Joseph G. and Chen, Ming-Hui and Sinha, Debajyoti},
  year      = {2001},
  publisher = {Springer},
  address   = {New York}
}

@article{jewson2018principles,
  title={Principles of Bayesian inference using general divergence criteria},
  author={Jewson, Jack and Smith, Jim Q and Holmes, Chris},
  journal={Entropy},
  volume={20},
  number={6},
  pages={442},
  year={2018}
}

@article{kim2003bayesian,
  title={Bayesian bootstrap for proportional hazards models},
  author={Kim, Yongdai and Lee, Jaeyong},
  journal={The Annals of Statistics},
  volume={31},
  number={6},
  pages={1905--1922},
  year={2003},
  publisher={Institute of Mathematical Statistics}
}

@book{klein2003survival,
  title={Survival analysis: techniques for censored and truncated data},
  author={Klein, John P and Moeschberger, Melvin L},
  volume={1230},
  year={2003},
  publisher={Springer}
}

@article{kleijn2012bernstein,
  title={The Bernstein-von-Mises theorem under misspecification},
  author={Kleijn, Bas JK and Van der Vaart, Aad W},
  journal={Electronic J. Statistics},
  volume={6},
  pages={354--381},
  year={2012}
}

@article{lee2016hierarchical,
  title={Hierarchical models for semicompeting risks data with application to quality of end-of-life care for pancreatic cancer},
  author={Lee, Kyu Ha and Dominici, Francesca and Schrag, Deborah and Haneuse, Sebastien},
  journal={Journal of the American Statistical Association},
  volume={111},
  number={515},
  pages={1075--1095},
  year={2016},
  publisher={Taylor \& Francis}
}

@article{lo1987large,
  title={A large sample study of the Bayesian bootstrap},
  author={Lo, Albert Y},
  journal={The Annals of Statistics},
  volume={15},
  number={1},
  pages={360--375},
  year={1987},
  publisher={JSTOR}
}

@article{lo1993bayesian,
  title={A Bayesian bootstrap for censored data},
  author={Lo, Albert Y},
  journal={The Annals of Statistics},
  volume={21},
  number={1},
  pages={100--123},
  year={1993},
  publisher={JSTOR}
}

@article{lyddon2019general,
  title={General Bayesian updating and the loss-likelihood bootstrap},
  author={Lyddon, Simon P and Holmes, Chris C and Walker, Stephen G},
  journal={Biometrika},
  volume={106},
  number={2},
  pages={465--478},
  year={2019},
  publisher={Oxford University Press}
}

@article{matsubara2024generalized,
  title={Generalized Bayesian inference for discrete intractable likelihood},
  author={Matsubara, Takuo and Knoblauch, Jeremias and Briol, Fran{\c{c}}ois-Xavier and Oates, Chris J},
  journal={Journal of the American Statistical Association},
  volume={119},
  number={547},
  pages={2345--2355},
  year={2024},
  publisher={Taylor \& Francis}
}

@article{mcgree2025approach,
  title={An Approach to Design Adaptive Clinical Trials With Time-to-Event Outcomes Based on a General Bayesian Posterior Distribution},
  author={McGree, James M and Overstall, Antony M and Jones, Mark and Mahar, Robert K},
  journal={Statistics in Medicine},
  volume={44},
  number={23-24},
  pages={e70207},
  year={2025},
  publisher={Wiley Online Library}
}

@inproceedings{minsker2014scalable,
  title={Scalable and robust Bayesian inference via the median posterior},
  author={Minsker, Stanislav and Srivastava, Sanvesh and Lin, Lizhen and Dunson, David},
  booktitle={International conference on machine learning},
  pages={1656--1664},
  year={2014},
  organization={PMLR}
}

@article{murray2015flexible,
  title={Flexible Bayesian survival modeling with semiparametric time-dependent and shape-restricted covariate effects},
  author={Murray, Thomas A and Hobbs, Brian P and Sargent, Daniel J and Carlin, Bradley P},
  journal={Bayesian analysis (Online)},
  volume={11},
  number={2},
  pages={381},
  year={2015}
}

@article{newton2021weighted,
  title={Weighted Bayesian bootstrap for scalable posterior distributions},
  author={Newton, Michael A and Polson, Nicholas G and Xu, Jianeng},
  journal={Canadian Journal of Statistics},
  volume={49},
  number={2},
  pages={421--437},
  year={2021},
  publisher={Wiley Online Library}
}

@article{newton1994approximate,
  title={Approximate Bayesian inference with the weighted likelihood bootstrap},
  author={Newton, Michael A and Raftery, Adrian E},
  journal={Journal of the Royal Statistical Society Series B: Statistical Methodology},
  volume={56},
  number={1},
  pages={3--26},
  year={1994},
  publisher={Oxford University Press}
}

@article{nie2023deep,
  title={Deep bootstrap for Bayesian inference},
  author={Nie, Lizhen and Ro{\v{c}}kov{\'a}, Veronika},
  journal={Philosophical Transactions of the Royal Society A: Mathematical, Physical and Engineering Sciences},
  volume={381},
  number={2247},
  pages={1--21},
  year={2023},
  publisher={The Royal Society}
}

@article{overstall2025gibbs,
  title={Gibbs optimal design of experiments},
  author={Overstall, Antony M and Holloway-Brown, Jacinta and McGree, James M},
  journal={Statistical Science},
  volume={0},
  number={0},
  pages={1--29},
  year={2025},
  publisher={Accepted}  
}

@article{prunster2023kidney,
  title={Kidney Donor Profile Index and allograft outcomes: interactive effects of estimated post-transplant survival score and ischaemic time},
  author={Prunster, Janelle and Wong, Germaine and Larkins, Nicholas and Wyburn, Kate and Francis, Ross and Mulley, William R and Ooi, Esther and Pilmore, Helen and Davies, Christopher E and Lim, Wai H},
  journal={Clinical Kidney Journal},
  volume={16},
  number={3},
  pages={473--483},
  year={2023},
  publisher={Oxford University Press}
}

@inproceedings{ranganath2014black,
  title={Black box variational inference},
  author={Ranganath, Rajesh and Gerrish, Sean and Blei, David},
  booktitle={Artificial intelligence and statistics},
  pages={814--822},
  year={2014},
  organization={PMLR}
}

@article{rao2009comprehensive,
  title={A comprehensive risk quantification score for deceased donor kidneys: the kidney donor risk index},
  author={Rao, Panduranga S and Schaubel, Douglas E and Guidinger, Mary K and Andreoni, Kenneth A and Wolfe, Robert A and Merion, Robert M and Port, Friedrich K and Sung, Randall S},
  journal={Transplantation},
  volume={88},
  number={2},
  pages={231--236},
  year={2009},
  publisher={LWW}
}

@article{rubin1981bayesian,
  title={The bayesian bootstrap},
  author={Rubin, Donald B},
  journal={The annals of statistics},
  volume={9},
  number={1},
  pages={130--134},
  year={1981},
  publisher={JSTOR}
}

@book{spiegelhalter2004bayesian,
  title={Bayesian approaches to clinical trials and health-care evaluation},
  author={Spiegelhalter, David J and Abrams, Keith R and Myles, Jonathan P},
  year={2004},
  publisher={John Wiley \& Sons}
}

@article{syring2019calibrating,
  title={Calibrating general posterior credible regions},
  author={Syring, Nicholas and Martin, Ryan},
  journal={Biometrika},
  volume={106},
  number={2},
  pages={479--486},
  year={2019},
  publisher={Oxford University Press}
}

@book{van1998asymptotic,
  title={Asymptotic statistics},
  author={Van der Vaart, Aad W},
  volume={3},
  year={1998},
  publisher={Cambridge university press}
}

@article{wu2023comparison,
  title={A comparison of learning rate selection methods in generalized Bayesian inference},
  author={Wu, Pei-Shien and Martin, Ryan},
  journal={Bayesian Analysis},
  volume={18},
  number={1},
  pages={105--132},
  year={2023},
  publisher={International Society for Bayesian Analysis}
}

@article{zhang2006from,
  title={From epsilon-entropy to KL-entropy: analysis of minimum risk},
  author={Zhang, Tong},
  journal={The Annals of Statistics},
  volume={34},
  number={1},
  pages={198--238},
  year={2006},
  publisher={Institute of Mathematical Statistics}
}

@article{zhou2023transformation,
  title={Transformation-invariant learning of optimal individualized decision rules with time-to-event outcomes},
  author={Zhou, Yu and Wang, Lan and Song, Rui and Zhao, Tuoyi},
  journal={Journal of the American Statistical Association},
  volume={118},
  number={544},
  pages={2632--2644},
  year={2023},
  publisher={Taylor \& Francis}
}

@article{zhu2025survival,
  title={Survival benefits of deceased donor kidney transplant vs waitlisting},
  author={Zhu, Lin and Teixeira-Pinto, Armando and Gately, Ryan and Boroumand, Farzaneh and Bakar, K Shuvo and Sabanayagam, Dharshana and Stanaway, Fiona F and Lim, Wai H and Wong, Germaine},
  journal={JAMA internal medicine},
  volume={185},
  number={12},
  pages={1471--1478},
  year={2025},
  publisher={American Medical Association}
}

\newpage 

\renewcommand{\thepage}{S\arabic{page}} 
\renewcommand{\thesection}{S\arabic{section}}  
\renewcommand{\thetable}{S\arabic{table}}  
\renewcommand{\thefigure}{S\arabic{figure}}
\setcounter{page}{0}
\setcounter{section}{0}
\setcounter{figure}{0}
\setcounter{table}{0}

\section*{Supplementary Material}

\section{The \texttt{R} package \texttt{BayesBoots} for BB-Cox}

\begin{verbatim}
## BB-Cox Example (lung cancer data)

> library(devtools)
> devtools::install_github("ksbakar/BayesBoots")
* installing *source* package 'BayesBoots' ...
** this is package 'BayesBoots' version '1.0.1'
** using staged installation
** R
** byte-compile and prepare package for lazy loading
** help
*** installing help indices
** building package indices
** testing if installed package can be loaded from temporary location
** testing if installed package can be loaded from final location
** testing if installed package keeps a record of temporary installation path
* DONE (BayesBoots)

> library(BayesBoots)
> library(survival)
> head(lung)
  inst time status age sex ph.ecog ph.karno pat.karno meal.cal wt.loss
1    3  306      2  74   1       1       90       100     1175      NA
2    3  455      2  68   1       0       90        90     1225      15
3    3 1010      1  56   1       0       90        90       NA      15
4    5  210      2  57   1       1       90        60     1150      11
5    1  883      2  60   1       0      100        90       NA       0
6   12 1022      1  74   1       1       50        80      513       0

> coxdata <- na.omit(lung)
> fit_cox <- runBB(
+  data = coxdata,
+  model = "BB-Cox",
+  formula = survival::Surv(time, status) ~ age + sex + ph.ecog + wt.loss,
+  prior_fn = prior_ridge
+)
> fit_cox
############################################
 Bayesian Bootstrap Induced Survival Models 
############################################

Model class: BB-Cox 

Call:
survival::Surv(time, status) ~ age + sex + ph.ecog + wt.loss

############################################

> summary(fit_cox)
############################################
 Bayesian Bootstrap Induced Survival Models 
############################################

Model type: BB-Cox
Call:
  survival::Surv(time, status) ~ age + sex + ph.ecog + wt.loss 

 BB Posterior estimates:
     term       beta     HR HR_lower HR_upper
1     age -0.0003768 0.9996   0.9742   1.0430
2     sex -0.4462627 0.6400   0.3819   0.8196
3 ph.ecog  0.5169200 1.6769   1.3340   2.5382
4 wt.loss -0.0099590 0.9901   0.9747   1.0148

> plot(fit_cox)
\end{verbatim}
\begin{figure}[h]
\centering
\includegraphics[width=0.7\textwidth]{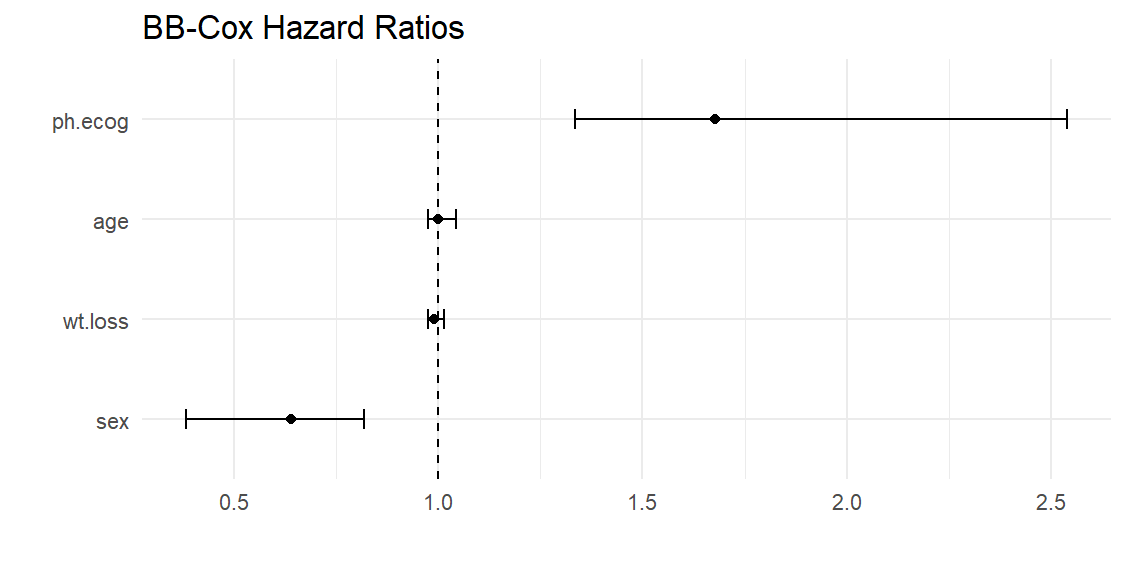}
\end{figure}
\begin{verbatim}
> plot(fit_cox, type="surv", group_var = "sex")    
\end{verbatim}
\begin{figure}[h]
\centering
\includegraphics[width=0.7\textwidth]{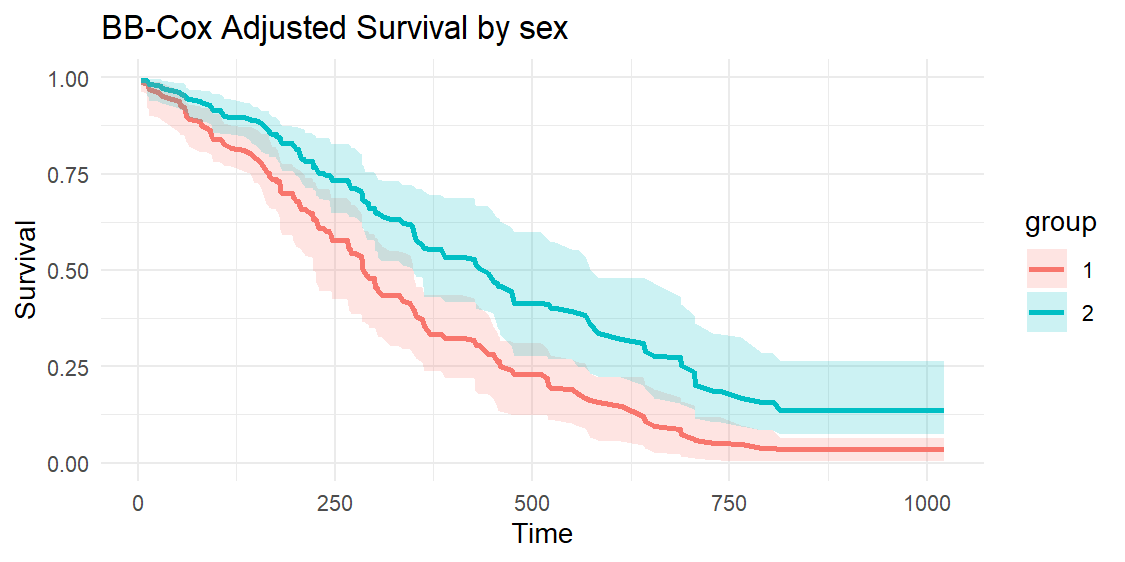}
\end{figure}

\newpage
 
\section{Simulation results: Validation statistics}

\begingroup\fontsize{8}{10}\selectfont

\begin{longtable}[t]{lrlrrrrrrr}
\caption{\label{tab:tab:performance_all_n}Performance of the Cox, Bayesian Cox, and BB-Cox methods across simulation scenarios and sample sizes.}\\
\toprule
\textbf{Scenario} & \textbf{Sample Size} & \textbf{Method} & \textbf{Bias} & \textbf{MSE} & \textbf{MAE} & \textbf{RMSE} & \textbf{Estimate} & \textbf{Lower} & \textbf{Upper}\\
\midrule
\endfirsthead
\caption[]{Performance of the Cox, Bayesian Cox, and BB-Cox methods across simulation scenarios and sample sizes. \textit{(continued)}}\\
\toprule
\textbf{Scenario} & \textbf{Sample Size} & \textbf{Method} & \textbf{Bias} & \textbf{MSE} & \textbf{MAE} & \textbf{RMSE} & \textbf{Estimate} & \textbf{Lower} & \textbf{Upper}\\
\midrule
\endhead
\endfoot
\bottomrule
\endlastfoot
Scenario-1 & 50 & Cox & 0.028 & 0.148 & 0.308 & 0.384 & 0.728 & 0.020 & 1.436\\
Scenario-1 & 50 & Bayesian Cox & 0.050 & 0.155 & 0.315 & 0.394 & 0.750 & 0.055 & 1.471\\
Scenario-1 & 50 & BB-Cox & 0.042 & 0.158 & 0.318 & 0.398 & 0.742 & 0.073 & 1.450\\
Scenario-2 & 50 & Cox & 0.024 & 0.139 & 0.293 & 0.373 & 0.724 & 0.016 & 1.431\\
Scenario-2 & 50 & Bayesian Cox & 0.010 & 0.129 & 0.283 & 0.360 & 0.710 & 0.033 & 1.411\\
Scenario-2 & 50 & BB-Cox & 0.010 & 0.132 & 0.287 & 0.364 & 0.710 & 0.055 & 1.390\\
Scenario-3 & 50 & Cox & -0.012 & 0.127 & 0.283 & 0.357 & 0.688 & -0.015 & 1.391\\
Scenario-3 & 50 & Bayesian Cox & 0.001 & 0.081 & 0.227 & 0.284 & 0.701 & 0.088 & 1.333\\
Scenario-3 & 50 & BB-Cox & -0.004 & 0.085 & 0.231 & 0.291 & 0.696 & 0.099 & 1.315\\
Scenario-1 & 100 & Cox & 0.025 & 0.063 & 0.198 & 0.251 & 0.725 & 0.237 & 1.212\\
Scenario-1 & 100 & Bayesian Cox & 0.036 & 0.065 & 0.201 & 0.256 & 0.736 & 0.259 & 1.224\\
Scenario-1 & 100 & BB-Cox & 0.032 & 0.066 & 0.202 & 0.257 & 0.732 & 0.262 & 1.218\\
Scenario-2 & 100 & Cox & 0.001 & 0.061 & 0.196 & 0.247 & 0.701 & 0.215 & 1.187\\
Scenario-2 & 100 & Bayesian Cox & -0.003 & 0.059 & 0.193 & 0.243 & 0.697 & 0.228 & 1.178\\
Scenario-2 & 100 & BB-Cox & -0.004 & 0.060 & 0.195 & 0.246 & 0.696 & 0.241 & 1.173\\
Scenario-3 & 100 & Cox & 0.016 & 0.061 & 0.198 & 0.248 & 0.716 & 0.228 & 1.204\\
Scenario-3 & 100 & Bayesian Cox & 0.022 & 0.050 & 0.178 & 0.223 & 0.722 & 0.273 & 1.180\\
Scenario-3 & 100 & BB-Cox & 0.021 & 0.051 & 0.180 & 0.225 & 0.721 & 0.276 & 1.174\\
Scenario-1 & 150 & Cox & -0.001 & 0.041 & 0.164 & 0.203 & 0.699 & 0.305 & 1.093\\
Scenario-1 & 150 & Bayesian Cox & 0.006 & 0.042 & 0.165 & 0.205 & 0.706 & 0.322 & 1.099\\
Scenario-1 & 150 & BB-Cox & 0.004 & 0.042 & 0.165 & 0.205 & 0.704 & 0.321 & 1.097\\
Scenario-2 & 150 & Cox & 0.009 & 0.043 & 0.166 & 0.207 & 0.709 & 0.314 & 1.104\\
Scenario-2 & 150 & Bayesian Cox & 0.007 & 0.042 & 0.164 & 0.205 & 0.707 & 0.324 & 1.098\\
Scenario-2 & 150 & BB-Cox & 0.005 & 0.043 & 0.167 & 0.207 & 0.705 & 0.331 & 1.099\\
Scenario-3 & 150 & Cox & -0.001 & 0.039 & 0.157 & 0.197 & 0.699 & 0.305 & 1.092\\
Scenario-3 & 150 & Bayesian Cox & 0.005 & 0.033 & 0.145 & 0.183 & 0.705 & 0.336 & 1.081\\
Scenario-3 & 150 & BB-Cox & 0.002 & 0.034 & 0.147 & 0.184 & 0.702 & 0.333 & 1.077\\
Scenario-1 & 200 & Cox & 0.010 & 0.031 & 0.142 & 0.176 & 0.710 & 0.370 & 1.050\\
Scenario-1 & 200 & Bayesian Cox & 0.015 & 0.032 & 0.143 & 0.178 & 0.715 & 0.383 & 1.054\\
Scenario-1 & 200 & BB-Cox & 0.014 & 0.032 & 0.143 & 0.179 & 0.714 & 0.383 & 1.052\\
Scenario-2 & 200 & Cox & 0.012 & 0.032 & 0.143 & 0.179 & 0.712 & 0.371 & 1.052\\
Scenario-2 & 200 & Bayesian Cox & 0.010 & 0.032 & 0.142 & 0.178 & 0.710 & 0.381 & 1.047\\
Scenario-2 & 200 & BB-Cox & 0.009 & 0.032 & 0.143 & 0.179 & 0.709 & 0.385 & 1.053\\
Scenario-3 & 200 & Cox & -0.000 & 0.027 & 0.129 & 0.163 & 0.700 & 0.359 & 1.040\\
Scenario-3 & 200 & Bayesian Cox & 0.005 & 0.024 & 0.122 & 0.154 & 0.705 & 0.384 & 1.033\\
Scenario-3 & 200 & BB-Cox & 0.003 & 0.024 & 0.123 & 0.155 & 0.703 & 0.378 & 1.032\\*
\end{longtable}
\endgroup{}

\end{document}